\documentclass[a4paper,12pt]{amsart}

\usepackage{amsmath}
\usepackage{amssymb}
\usepackage{amsthm}
\usepackage{hyperref}
\usepackage[a4paper]{geometry}
\usepackage{microtype}
\usepackage{todonotes}
\usepackage{enumerate}   
\usepackage{enumitem}
\usepackage{cite}
\DeclareFontFamily{U}{mathx}{\hyphenchar\font45}
\DeclareFontShape{U}{mathx}{m}{n}{
      <5> <6> <7> <8> <9> <10>
      <10.95> <12> <14.4> <17.28> <20.74> <24.88>
      mathx10
      }{}
\DeclareSymbolFont{mathx}{U}{mathx}{m}{n}
\DeclareFontSubstitution{U}{mathx}{m}{n}
\DeclareMathAccent{\widecheck}{0}{mathx}{"71}

\theoremstyle{plain}
\newtheorem{thm}{Theorem}[section]
\newtheorem{cor}{Corollary}[thm]
\newtheorem{lem}[thm]{Lemma}
\newtheorem{prop}[thm]{Proposition}

\theoremstyle{definition}

\newtheorem{exmp}{Example}[section]

\theoremstyle{remark}
\newtheorem*{rem}{Remark}

\newcommand\inner[2]{\langle #1, #2 \rangle}

\newcommand{\tco}{\mathcal{T}^1}
\newcommand{\bo}{B(L^2(\reals^d))}
\newcommand{\reals}{\mathbb{R}}

\newcommand{\HS}{L^2(\reals^d)}
\newcommand{\F}{\mathcal{F}}
\newcommand{\SC}{\mathcal{T}}
\newcommand{\tr}{\mathrm{tr}}

\begin{document}
\pagestyle{plain}
\title{Convolutions for Berezin quantization and Berezin-Lieb inequalities}
\author{Franz Luef}
\author{Eirik Skrettingland} 
\address{Department of Mathematics\\ NTNU Norwegian University of Science and
Technology\\ NO–7491 Trondheim\\Norway}
\email{franz.luef@math.ntnu.no, eirikskr@stud.ntnu.no}
\keywords{Berezin quantization, Berezin-Lieb inequalities, Fourier-Wigner transform, phase space representations }
\subjclass{47G30; 35S05; 46E35; 47B10}
\begin{abstract}
Concepts and results from quantum harmonic analysis, such as the convolution between functions and operators or between two operators, are  
identified as the appropriate setting for Berezin quantization and Berezin-Lieb inequalities. Based on this insight we provide a rigorous approach to generalized phase-space representation introduced by Klauder-Skagerstam and their variants of Berezin-Lieb inequalities in this setting. Hence our presentation of the results of Klauder-Skagerstam gives a more conceptual framework, which yields as a byproduct an interesting perspective on the connection between Berezin quantization and Weyl quantization. 
\end{abstract}
\maketitle \pagestyle{myheadings} \markboth{F. Luef and E. Skrettingland}{Convolutions for Berezin quantization}
\thispagestyle{empty}
\section{Introduction}
Phase-space representations are of interest for problems in quantum mechanics but also in time-frequency analysis and mathematics. The Wigner distribution is the most-studied phase space representation
\[ W(\psi,\psi)(x,\omega)=\int_{\mathbb{R}^d}\psi(x+\tfrac{t}{2})\overline{\psi(x-\tfrac{t}{2})}e^{-2\pi i\omega t}dt.\]
There are many other representation of this type, e.g. the short-time Fourier transform, the Born-Jordan distribution, Cohen's class, the Husimi representation and the Glauber-Sudarshan representation. There are various names for the Husimi representation and the Glauber-Sudarshan representation. We collected some of them in the next table.
\begin{table} 
  \begin{tabular}{p{7cm}p{7cm}} \hline
    \textbf{This paper} & \textbf{Mathematical physics} \\
    \hline
    Time-frequency shift $\pi(z)$ & Weyl operator $e^{2\pi i\omega Q-ixP}$ \cite{Klauder:2007,Kiukas:2012gt} \\ 
     Integrated Schr{\"o}dinger representation & Weyl quantization \cite{Klauder:2007} \\ 
     Twisted Weyl symbol, Fourier-Wigner transform & Weyl representation \cite{Klauder:2007} \\ 
     Glauber-Sudarshan representation & anti-Wick symbol \cite{Folland:1989tl}, contravariant Berezin symbol, upper symbol \cite{Simon:1980}, symbol for localization operator \cite{Bayer:2014td} \\ 
     Husimi representation & Berezin transform \cite{Bayer:2014td}, covariant Berezin symbol, lower symbol \cite{Simon:1980} \\
  \hline \end{tabular}
  \caption{A dictionary relating the terminology in this paper to other common terminologies in mathematical physics.} 
\end{table}
The starting point of Klauder and Skagerstam \cite{Klauder:2007} is the identification of the Husimi and Glauber-Sudarshan representations as objects build out of rank-one operators given by Gaussians. Then the main idea is to replace these rank-one operators by some density matrix or more generally a trace class operator $\sigma$. The associated phase space representations $S_\sigma$ and $S_{-\sigma}$ are the aforementioned generalized phase-space representations. 

The main goal of this paper is to develop the theory of these phase-space representations rigorously. The theory of Werner \cite{Werner:1984} seems to be a convenient setting for the results in \cite{Klauder:2007} since the convolution between a function and an operator, and between two operators, is precisely how Klauder-Skagerstam define their generalized phase-space representation. In their definition they assume that the Fourier-Wigner transform is non-zero on phase space. We link this condition to Werner's Tauberian theorems and show that it is equivalent to a desirable property: Any bounded operator on $L^2(\mathbb{R}^d)$ is uniquely determined by its generalized Husimi representation. In addition, we note that one may define the phase space representations of Klauder-Skagerstam even in the case where the Fourier-Wigner transform has some zeros. Finally, we show that any trace class operator may be expressed in terms of a generalized phase-space representation which is a consequence of Cohen-Hewitt's theorem for Banach modules \cite{Luef:2017vs}.

All these results are based on the interplay between the Berezin quantization and the Weyl quantization which is not so much stressed in the existing literature, except for \cite{Landsman:1999}.  In the paper \cite{Klauder:2011} generalizations of the Berezin-Lieb inequalities are established in terms of the generalized phase space representations. We state a more general version of these inequalities and point out that these have been first proved by Werner in \cite{Werner:1984}. 

The proofs of the properties of the generalized phase space representations are done in the framework of \cite{Werner:1984} and they are expressed in terms of convolutions. Consequently, we are in the position to use the Fourier-Wigner transform and other facts from \cite{Werner:1984,Luef:2017vs} in the proofs of the main results of \cite{Klauder:2011}. The presentation of \cite{Werner:1984} in terms of time-frequency analysis in \cite{Luef:2017vs} sheds some light on a remark in \cite{Klauder:2007} on a potential link in this setting. The paper \cite{Werner:1984} builds on the work of Holevo \cite{Holevo:1979}, and the later results have later been extended in the context of covariant observables, which are a special class of positive operator valued measures, by \cite{Busch:1997,Cassinelli:2003,Kiukas:2006,Kiukas:2006a} and the monograph \cite{Busch:2016}. We indicate in the final section how to adapt the definitions of Klauder-Skagerstam to the setting of (unimodular) locally compact groups using results from \cite{Kiukas:2006}.

We close this introduction with a brief discussion of results related to generalized phase space representations in the setting of operational quantum physics. The simultaneous measurement of non-commuting observables using some measurement device necessarily leads to the introduction of some noise, by Heisenberg's uncertainty principle.  Wo\'dkiewicz \cite{Wodkiewicz:1984} has observed that generalized Husimi representations incorporate this noise due to the measuring device and considered a version of the Husimi representation in an attempt to define a phase-space representation of states that may realistically be recorded by experiments. W\"unsche and Buzek \cite{Buzek:1997} followed this line of research and discussed if the generalized Husimi function contains complete information about the measurement of a state and propose a formal reconstruction procedure which asks about the ``inversion" of the generalized Husimi symbol. A consequence of our results shows that this procedure will not always work. This connection and result has also been noted in \cite{Kiukas:2012gt,Busch:2016}.

\section{Prerequisites}
\subsection{Notation and conventions} \label{sec:notation} 
Elements of $\reals^{2d}$ will often be written in the form $z=(x,\omega)$ for $x,\omega\in \reals^d$ and $[z,z^\prime]$ is the standard symplectic form $[z,z^\prime]=\omega_1\cdot x_2-\omega_2 \cdot x_1$ of $z=(x_1,\omega_1)$ and $z^\prime=(x_2, \omega_2)$. For two elements $\xi, \eta$ in some Hilbert space $\mathcal{H}$, we define the operator $\xi \otimes \eta$ on $\mathcal{H}$ by $\xi \otimes \eta (\zeta)=\inner{\zeta}{\eta}\xi$, where $\zeta \in \mathcal{H}$ and $\inner{\cdot}{\cdot}$ is the inner product on $\mathcal{H}$. 
We introduce the parity operator $P$ by  $\check{\psi}(x)=P\psi(x)=\psi(-x)$ for any $x\in \reals^d$ and $\psi:\reals^d\to \mathbb{C}$, and define $\psi^*$ by $\psi^*(x)=\overline{\psi(x)}$.
For $p<\infty$, $\SC^p$ will denote the Schatten $p$-class of operators on the Hilbert space $\HS$ with singular values in $\ell^p$. Furthermore, we define $\SC^{\infty}$ to be $\bo$; all the bounded, linear operators on $L^2(\mathbb{R}^d)$.

\subsection{Time-frequency shifts or Weyl operators} \label{sec:tfshift} 
 If $\psi:\reals^d\to \mathbb{C}$ and $z=(x,\omega)\in \reals^{2d}$, we define the \textit{translation operator} $T_x$ by $T_x\psi(t)=\psi(t-x)$, the \textit{modulation operator} $M_{\omega}$ by $M_{\omega}\psi(t)=e^{2 \pi i \omega \cdot t} \psi(t)$ and the \textit{time-frequency shifts} $\pi(z)$ by $\pi(z)=M_{\omega}T_x$. The time-frequency shifts are fundamental in time-frequency analysis, but they also play an important role in quantum mechanics as $\pi(x,\omega)=e^{2\pi i\omega Q-ixP}$ where $Q$ and $P$ are the canonical operators satisfying the uncertainty relation $\left[Q,P\right]=i$ (we let $\hbar=1$)\cite{Hall:2013vm}. In physics these operators are sometimes called the \textit{Weyl operators} $U(x,\omega)$ in \cite{Kiukas:2012gt,Klauder:2007}. 
For $\psi,\phi\in L^2(\reals^d)$ the \textit{short-time Fourier transform} (STFT) $V_{\phi}\psi$ of $\psi$ with window $\phi$ is the function on $\reals^{2d}$
 defined by
 \begin{equation*}
  V_{\phi}\psi(z)=\inner{\psi}{\pi(z)\phi}
\end{equation*}
for $z\in \reals^{2d}$. 
\subsection{The integrated Schr{\"o}dinger representation and twisted convolution} \label{sec:intschrodrep}
 The \textit{integrated Schr{\"o}dinger representation} is the map $\rho: L^1(\reals^{2d})\to \bo$ given by 
	\begin{equation*}
  \rho(f)=\iint_{\reals^{2d}} f(z) e^{-\pi i x \cdot \omega} \pi(z) \ dz,
\end{equation*}
where the integral is defined in the weak sense by requiring that 
\begin{equation*}
  \inner{\eta}{\rho(f)\psi}= \iint_{\reals^{2d}} f(z) e^{-\pi i x \cdot \omega} \inner{\eta}{\pi(z)\psi} \ dz
\end{equation*}
for any $\psi, \eta \in \HS$ (we refer to \cite{Luef:2017vs} for a detailed exposition of weak integration in this context). We say that $f$ is the \textit{twisted Weyl symbol} of $\rho(f)$. In the physics literature this is sometimes referred to as the Weyl representation of an operator \cite{Klauder:2007}, and the map $f\mapsto \rho(f)$ is a version of the \textit{Weyl quantization}.

We will use the important product formula $\rho(f)\rho(g)=\rho(f \natural g)$, where the product $\natural$ is the \textit{twisted convolution}, defined by 
$
  f\natural g (z) = \iint_{\reals^{2d}} f(z-z')g(z') e^{\pi i [z,z']} \ dz'
$
for $f,g \in L^1(\reals^{2d})$  \cite{Folland:1989tl,Grochenig:2001}.

For this paper it is essential that $\rho$ may be extended to a unitary operator from $L^2(\reals^{2d})$ to $\SC^2$, and that the twisted convolution  $f\natural g$ may be defined for $f,g \in L^2(\reals^{2d})$ with norm estimate $\|f\natural g\|_{L^2}\leq \|f\|_{L^2}\|g\|_{L^2}$, see \cite{Folland:1989tl}.

\subsection{Berezin quantization: Husimi and Glauber-Sudarshan representations}
Let $\phi$ be the Gaussian $\phi(x)=2^{d/4}e^{-\pi x\cdot x}$. If $A\in \bo$ can be represented using a vector-valued integral as 
\begin{equation} \label{eq:glaubersudarshan}
  A\psi= \iint_{\reals^{2d}} A_{G-S}(z) \cdot V_{\phi}\psi(z) \pi(z) \phi \ dz
\end{equation}
for some function $A_{G-S}$ on $\reals^{2d}$, we call $A_{G-S}$ the \textit{Glauber-Sudarshan representation} of $A$\cite{Klauder:2007}. On the other hand, the mapping $A_{G-S}\mapsto A$, with $A$ defined by equation \eqref{eq:glaubersudarshan}, defines a quantization procedure known as \textit{Berezin quantization} \cite{Landsman:1999}. In this sense $A\mapsto A_{G-S}$ is the inverse of Berezin quantization.
Furthermore, the \textit{Husimi representation} of $A\in \bo$ is the function $A_H$ given by
\begin{equation*}
  A_H(z)=\inner{A\pi(z)\phi}{\pi(z)\phi}
\end{equation*}
for $z\in \reals^{2d}$\cite{Klauder:2007}, also known as the Berezin transform of $A$ \cite{Bayer:2014td}. 
\subsection{Convolutions of operators and functions} \label{sec:werner}
The mains tools of this paper will be provided by the theory of convolutions of operators and functions due to Werner \cite{Werner:1984}. In order to introduce these convolution operations, we will first need to define a shift for operators. For $z\in \reals^{2d}$ and $A\in \bo$, we define the operator $\alpha_z(A)$ by
\begin{equation*}
  \alpha_z(A)=\pi(z)A\pi(z)^*.
\end{equation*}
It is easily confirmed that $\alpha_z\alpha_{z'}=\alpha_{z+z'}$, and we will informally think of $\alpha$ as a shift or translation of operators. 

Similarly we define the analogue of the involution $f\mapsto \check{f}$ of a function, for an operator $A\in \bo$ by
\begin{equation*}
  \check{A}=PAP,
\end{equation*}
where $P$ is the parity operator. 

Using $\alpha$, Werner defined a convolution operation between functions and operators \cite{Werner:1984}. If $f \in L^1(\reals^{2d})$ and $S \in \tco$ we define the \textit{operator} $f\ast S$ by
\begin{equation} \label{eq:operatorconv}
  f\ast S := S\ast f = \iint_{\reals^{2d}}f(y)\alpha_y(S) \ dy
\end{equation}
where the integral is interpreted in the weak sense. Then $f\ast S \in \tco$ and $\|f\ast S \|_{L^1}\leq \|f \|_{L^1}\|S\|_{\tco}$ \cite[Prop. 2.5]{Luef:2017vs}.

For two operators $S,T \in \tco$, Werner defined the \textit{function} $S\ast T$ by 
\begin{equation} \label{eq:convolutionoftwooperators}
  S \ast T(z) = \tr(S\alpha_z (\check{T}))
\end{equation}
for $z\in \reals^{2d}$. The following result shows that $S\ast T\in L^1(\reals^{2d})$ and provides an important formula for its integral \cite[Lem. 3.1]{Werner:1984}.
\begin{lem}
\label{lem:werner}
Let $S,T \in \tco$. The function $z \mapsto \tr(S\alpha_zT)$ for $z\in \reals^{2d}$ is integrable and $\|\tr(S\alpha_zT)\|_{L^1} \leq \|S\|_{\tco} \|T\|_{\tco}$.

Furthermore,
\begin{equation*}
	\iint_{\reals^{2d}} \tr(S\alpha_zT) \ dz = \tr(S)\tr(T).
\end{equation*}
\end{lem}

The convolutions can be defined on different $L^p$-spaces and Schatten $p$-classes \cite{Luef:2017vs,Werner:1984}.

\begin{prop} \label{prop:convschatten}
	Let $1\leq p,q,r \leq \infty$ be such that $\frac{1}{p}+\frac{1}{q}=1+\frac{1}{r}$. If $f\in L^p(\reals^{2d}), g\in L^q(\reals^{2d}), S \in \SC^p$ and $T\in \SC^q$. Then the following convolutions may be defined and satisfy the norm estimates
	\begin{align*}
  		\|f\ast T\|_{\SC^r} &\leq \|f\|_{L^p} \|T\|_{\SC^q}, \\
  		\|g\ast S\|_{\SC^r} &\leq \|g\|_{L^q} \|S\|_{\SC^p}, \\
  		\|S\ast T\|_{L^r} &\leq \|S\|_{\SC^p} \|T\|_{\SC^q}.  		
	\end{align*}
	\end{prop}

Since the convolutions between operators and functions can produce both operators and functions as output, the associativity of the convolutions is not trivial \cite{Luef:2017vs, Werner:1984}. 
\begin{prop} \label{prop:associative}
The convolution operations in proposition \ref{prop:convschatten} are commutative and associative.
\end{prop}
As was recently pointed out in \cite{Luef:2017vs}, the convolutions make the Schatten classes $\SC^p$ into \textit{Banach modules} over $L^1(\reals^{2d})$ if the module multiplication is defined by $(f,S)\mapsto f\ast S$ for $f\in L^1(\reals^{2d})$ and $S\in \SC^p$. By using the Cohen-Hewitt theorem for Banach modules \cite{Graven:1974}, one obtains that any operator in $\SC^p$ for $p<\infty$ can be written as a convolution\cite[Prop. 7.4]{Luef:2017vs}.
\begin{prop} \label{prop:banachmod}
	Given $T\in \SC^p$ for $p<\infty$, there exists $f\in L^1(\reals^{2d})$ and $S\in \SC^p$ such that $T=f\star S$.
\end{prop}

\subsection{Fourier transforms for functions and operators} \label{sec:fourier}
For functions $f\in L^1(\reals^{2d})$ we will use the \textit{symplectic Fourier transform} $\F f$, given by 
\begin{equation*}
\F f(z)=\iint_{\reals^{2d}} f(z') e^{-2 \pi i [z,z']} \ dz'
\end{equation*}
for $z\in \reals^{2d}$, where $[.,.]$ is the standard symplectic form $[(x_1,\omega_1),(x_2, \omega_2)]=\omega_1\cdot x_2-\omega_2 \cdot x_1$. 

 For operators $S\in \tco$, the \emph{Fourier-Wigner transform} $\F_WS$ of $S$ is the function given by
	\begin{equation*}
	\F_W S(z)=e^{-\pi i x \cdot \omega}\tr(\pi(-z)S)
	\end{equation*}
	for $z\in \reals^{2d}$. In the terminology of Werner \cite{Werner:1984,Kiukas:2012gt,Keyl:2015bn} this is the Fourier-Weyl transform, but we follow Folland \cite{Folland:1989tl} and call it the Fourier-Wigner transform.

The relationship between the Fourier-Wigner transform and the integrated Schr{\"o}dinger representation will play an important part in our considerations.

\begin{prop} \label{prop:wignerisunitary}
The Fourier-Wigner transform extends to a unitary operator $\F_W: \SC^2\to L^2(\reals^{2d})$. This extension is the inverse operator of the integrated Schr{\"o}dinger representation $\rho $, and 
\begin{equation} \label{eq:twistedconv}
\F_W(ST)=\F_W(S) \natural \F_W(T)
\end{equation}
 for $S,T \in \SC^2$.
\end{prop}

The Fourier-Wigner transform shares many properties with the regular Fourier transform\cite{Werner:1984, Luef:2017vs}, and in particular it interacts with the convolutions defined by Werner in the expected way .  

\begin{prop}
	\label{prop:convolutionandft}
	Let $f \in L^1(\reals^{2d})$ and $S,T \in \tco$. 
	\begin{enumerate}
		\item $\F (S \ast T)=\F_W(S)\F_W(T)$.
		\item $\F_W(f \ast S)=\F(f)\F_W(S)$.
	\end{enumerate}	
\end{prop}
This proposition is not merely aesthetically pleasing, but expresses the connection between Weyl quantization and Berezin quantization, as we will see in example \ref{exmp:berweyl}. 

\section{Generalized Husimi and Glauber-Sudarshan representations as convolutions}
In \cite{Klauder:2007}, Klauder and Skagerstam proposed extensions of the Husimi and Glauber-Sudarshan representations of operators. Given a fixed operator $\sigma\in \tco$ such that $\F_W(\sigma)$ has no zeros, they obtained for $S\in \tco$ a generalized Husimi representation $S_{\sigma}$ and a generalized Glauber-Sudarshan representation $S_{-\sigma}$. The representations $S_{\sigma}$ and $S_{-\sigma}$ generalize the Husimi and Glauber-Sudarshan representation in the sense that $S_{\sigma}=S_H$ and $S_{-\sigma}=S_{G-S}$ when $\sigma =\phi \otimes \phi$ for $\phi(x)=2^{d/4}e^{-\pi x\cdot x}$. The main result of this paper is that the generalized phase space representations may be described using Werner's convolutions. The next theorem makes this precise, and part (3) uses the associativity of convolutions to reprove a relation between $S_{\sigma}$ and $S_{-\sigma}$ due to Klauder and Skagerstam \cite{Klauder:2007}. In the proof we sketch the formal calculations used by Klauder and Skagerstam \cite{Klauder:2007} to obtain the generalized phase space representations. 

\begin{thm} \label{thm:main}
 Fix $\sigma\in \tco$ and let $S\in \tco$. 
 \begin{enumerate}
\item $S_{\sigma}=S\ast \check{\sigma}$.  
\item If $S=f\ast \sigma^*$ for some $f\in L^1(\reals^{2d})$, then $S_{-\sigma}=f$. 
\item If $S_{-\sigma}\in L^1(\reals^{2d})$ exists, then by parts (1) and (2) $$S_{\sigma}=(S_{-\sigma}\ast \sigma^*)\ast \check{\sigma}=S_{-\sigma}\ast (\sigma^*\ast \check{\sigma})=S_{-\sigma}\ast \tr(\sigma^* \alpha_z \sigma).$$
 \end{enumerate}
Furthermore, if $\sigma=\phi\otimes \phi$ where $\phi(x)=2^{d/4}e^{-\pi x\cdot x}$ for $x\in \reals^d$, then $S_{\sigma}=S_H$ and $S_{-\sigma}=S_{G-S}$.
 	
 \end{thm}
 
 \begin{proof}
 	To obtain the generalized representations, Klauder and Skagerstam started from the known relations
\begin{equation} \label{eq:basisforklauder}
  \tr(S^*T)=\iint_{\reals^{2d}} \overline{\F_W(S)(z')}\F_W(T)(z') \ dz'= \iint_{\reals^{2d}} \overline{\F(S_{G-S})(z')} \F( T_{H})(z') \ dz'
\end{equation}
for $S,T\in \tco$, which we will prove for the sake of completeness using Werner's convolutions in example \ref{exmp:traceandintegral}.
 They then fixed a $\sigma \in \tco$ such that $\F_W(\sigma)$ vanishes nowhere, in order to generalize $S_{G-S}$ to a representation $S_{-\sigma}$, and $T_{H}$ to a representation $T_{\sigma}$. They required that $S_{-\sigma}$ and $T_{\sigma}$ should satisfy the obvious generalization of \eqref{eq:basisforklauder}, and observed that this would hold if $S_{-\sigma}$ and $T_{\sigma}$ were introduced by the following formal calculation based on equation \eqref{eq:basisforklauder}:
 \begin{align*}
  \tr(S^*T)&=\iint_{\reals^{2d}} \left[\frac{\F_W(S)(z')}{\tr(\pi(z')^*\sigma^*)}\right]^* \tr(\pi(z')\sigma)\F_W(T)(z') \ dz' \\
  &:= \iint_{\reals^{2d}} \overline{\F(S_{-\sigma})(z')} \F(T_{\sigma})(z') \ dz' \\
  &=  \iint_{\reals^{2d}} \overline{S_{-\sigma}(z')} T_{\sigma}(z') \ dz',
\end{align*}
where we used Plancherel's theorem in the last step. One may then derive explicit expressions for $S_{-\sigma}$ and $T_{\sigma}$ \cite{Klauder:2007}. The generalized Husimi representation $T_{\sigma}$ is given by $T_{\sigma}(z)=\tr(T\alpha_z(\sigma))$, which clearly equals $T\ast \check{\sigma}$ from the definition of the convolution of operators. The most relevant expression for the generalized Glauber-Sudarshan representation $S_{-\sigma}$, is that when $S=\iint_{\reals^{2d}} f(z) \alpha_z{\sigma^*} \ dz$ for some $f\in L^1(\reals^{2d})$, then $f=S_{-\sigma}$. In other words, if $S=f\ast \sigma^*$, then $f=S_{-\sigma}$. This proves the first two parts of the theorem. Part (3) is a simple consequence of the first two parts and the associativity of Werner's convolutions. Finally, one may easily check that $T_{\sigma}=T_H$ and $S_{-\sigma}=S_{G-S}$ for $\sigma=\phi \otimes \phi$, where $\phi(x)=2^{d/4}e^{-\pi x\cdot x}$.
 \end{proof}
\begin{rem}
	As a simple calculation shows, $S_{\sigma}=S_{G-S}$ when $\sigma=\phi \otimes \phi$ for $\phi(x)=2^{d/4}e^{-\pi x\cdot x}$. Berezin quantization is given by the mapping $S_{G-S}\mapsto S$, and by theorem \ref{thm:main} this mapping could equivalently be written as $f\mapsto f\ast \left(\phi\otimes \phi\right)^*$. By replacing $\phi\otimes \phi$ with any $\sigma\in \tco$, we may for each $\sigma\in \tco$ define a version of Berezin quantization by $f\mapsto f\ast \sigma^*$.
\end{rem}
 Since theorem \ref{thm:main} relates the generalized phase space representations to Werner's convolutions, the theory from sections \ref{sec:werner} and \ref{sec:fourier} may be applied to these representations. From proposition \ref{prop:convschatten} we get precise conditions under which the generalized Husimi and Glauber-Sudarshan representations of an operator $S$ belong to various $L^p$-spaces, depending on which Schatten $p$-classes $S$ and $\sigma$ belong to.
  
  Theorem \ref{thm:main} also shows that $S_{\sigma}$ and $S_{-\sigma}$ may be defined using Werner's convolutions even when $\F_W(\sigma)$ has zeros. Furthermore, it implies that even if $\F_W (\sigma)$ has no zeros, not every operator $S\in \tco$ has a generalized Glauber-Sudarshan representation with respect to $\sigma$.
	\begin{prop}
	Assume $\sigma\in\tco$ and $\F_W (\sigma)$ has no zeros. Then not every operator $S\in \tco$ has a generalized Glauber-Sudarshan representation $S_{-\sigma}$ in $L^1(\reals^{2d})$. 
	\end{prop}
	\begin{proof}
	The operator $\sigma^*$ has no generalized Glauber-Sudarshan representation. If we assume that there is a $\sigma^*_{-\sigma}\in L^1(\reals^{2d})$ such that $\sigma^*=\sigma^*_{-\sigma} \ast \sigma^*$, then applying the Fourier-Wigner transform to both sides gives that $\F_W(\sigma^*)=\F(\sigma^*_{-\sigma})\F_W(\sigma^*)$, hence $\F (\sigma^*_{-\sigma})=1$, which contradicts the Riemann-Lebesgue lemma.
 \end{proof}
 On the other hand, the theory of Banach modules (Cohen-Hewitt's Theorem) allows us to conclude that any $S\in \tco$ has an integrable generalized Glauber-Sudarshan representation with respect to \textit{some} $\sigma\in \tco$.
 \begin{prop}
 	If $S\in \tco$, there exists some $\sigma \in \tco$ such that $S$ has a generalized Glauber-Sudarshan representation $S_{-\sigma}\in L^1(\reals^{2d})$.
 \end{prop}
 \begin{proof}
 	By proposition \ref{prop:banachmod}, there exists some $f\in L^1(\reals^{2d})$ such that $f\ast \sigma^*$, and by theorem \ref{thm:main} we get that $f=S_{-\sigma}$. 
 \end{proof}
 \begin{rem}
 	An example of a generalized Husimi-representation has recently been considered by Keller in \cite{Keller:2017}, where $\sigma$ is taken to be a finite-rank operator.   
 \end{rem}
 As mentioned in the introduction, the generalized Glauber-Sudarshan representation in \cite{Klauder:2007} has appeared in the literature before the publication of \cite{Klauder:2007}, see for instance chapter VI of \cite{Busch:1997} or \cite[Eq. 2.12]{Buzek:1997}. Hence one could similarly rephrase the phase-space representations discussed in those papers using Werner's convolutions, and in fact the connection between Werner's theory from \cite{Werner:1984} and \cite{Buzek:1997} has already been pointed out in \cite{Kiukas:2012gt}, which we will return to in section \ref{sec:tauberian}.
 \subsection{An operational approach to the generalized Husimi representation} \label{sec:operational}
 Wo\'dkiewicz \cite{Wodkiewicz:1984} considered a version of the Husimi representation in an attempt to define a phase-space representation of states that may realistically be recorded by experiments, and this approach also leads naturally to the introduction of the generalized Husimi functions, see \cite{Buzek:1997}. In particular, this allows us to interpret the generalized Husimi representation from an \textit{operational} point of view, different from the perspective adopted in \cite{Klauder:2007}. The simultaneous measurement of non-commuting observables using some measurement device necessarily leads to the introduction of some noise, by Heisenberg's uncertainty principle. In W\'odkiewicz' interpretation, the generalized Husimi representation incorporates this noise due to the measuring device: namely, if the system is in the state $S\in \tco$ and the measurement device is in the state $\sigma$, then we assign to $S$ the phase-space representation $S_{\sigma}=S\ast \check{\sigma}$. In particular, the regular Husimi representation $S_{H}$ corresponds to a measurement device in the ground state of the harmonic oscillator.

  \subsection{Using convolutions to reprove relations between Weyl and Berezin quantization} 
 By theorem \ref{thm:main} we may now use the theory of Werner's convolutions to shed light on relations between the Weyl, Husimi and Glauber-Sudarshan representations. Known relations between these representations can now be expressed neatly as the relations $\F_W(f\ast T)=\F(f)\F_W(T)$ and $\F(S\ast T)=\F_W(S)\F_W(T)$, or as the associativity of convolutions as in part (3) of theorem \ref{thm:main}.
 \begin{exmp} \label{exmp:traceandintegral}
 	 We will use this approach to prove equation \eqref{eq:basisforklauder}. By proposition \ref{prop:wignerisunitary}, $\F_W(S^*T)(z)=F_W(S^*)\natural \F_W(T)(z)$, and if we evaluate this equation at $z=0$, we get the first part of equation \eqref{eq:basisforklauder}. The second part follows from theorem \ref{thm:main}, which states that $S=S_{-\sigma}\ast \sigma^*$ if $S_{-\sigma}\in L^1(\reals^{2d})$ exists. Using $\F_W(S)=\F_W(S_{-\sigma}\ast \sigma^*)=\F(S_{-\sigma})\F_W(\sigma^*)$, we can write the first part of equation \eqref{eq:basisforklauder} as 
\begin{align*}
     \tr(S^*T)&=\iint_{\reals^{2d}} \overline{\F_W(S)(z')}\F_W(T)(z') \ dz' \\
     &= \iint_{\reals^{2d}} \overline{\F(S_{-\sigma})(z')\F_W(\sigma^*)(z')}\F_W(T)(z') \ dz' \\
     &= \iint_{\reals^{2d}} \overline{\F(S_{-\sigma})(z')}\F_W(\check{\sigma})(z')\F_W(T)(z') \ dz' \\
     &=\iint_{\reals^{2d}} \overline{\F S_{-\sigma}(z')}\F(T_{\sigma})(z') \ dz',
\end{align*}
where the last equality uses $\F_W(\check{\sigma})\F_W(T)=\F(T\ast \check{\sigma})=\F T_{\sigma}$ by theorem \ref{thm:main}. By picking $\sigma=\phi \otimes \phi$ with $\phi(x)=2^{d/4}e^{-\pi x\cdot x}$ we recover the second part of equation \eqref{eq:basisforklauder}.
  \end{exmp}

\begin{exmp} \label{exmp:berweyl}
 $\F_W(f\ast T)=\F(f)\F_W(T)$ provides a known link between Weyl quantization and Berezin quantization. Assume that $S\in \tco$ may be represented by $S=S_{-\sigma}\ast \sigma^*$, where $S_{-\sigma}\in L^1(\reals^{2d})$ and $\sigma \in \tco$. By proposition \ref{prop:wignerisunitary}, $\F_W(S)$ is the twisted Weyl symbol of $S$, and by proposition \ref{prop:convolutionandft}
  \begin{equation*}
  \F_W(S)=\F(S_{-\sigma})\F_W(\sigma^*).
\end{equation*}
 If $\sigma=\phi \otimes \phi$ with $\phi(x)=2^{d/4}e^{-\pi x\cdot x}$, one may calculate that $\F_W(\sigma^*)=e^{2\pi i x \cdot \omega}e^{\frac{\pi}{2}(z\cdot z)}$. In this case $S_{-\sigma}=S_{G-S}$, and we obtain the relation
\begin{equation*}
  \F(S_{G-S})=\F_W (S) e^{-2\pi i x \cdot \omega}e^{-\frac{\pi}{2}(z\cdot z)},
\end{equation*}
which relates the symbol of $S$ in Weyl quantization, $\F_W (S)$, to the symbol of $S$ in Berezin quantization, $S_{G-S}$.
Similar expressions for the Husimi representation are obtained in the same way. 
\end{exmp}
 \section{Berezin-Lieb inequalities}
 In \cite{Klauder:2011}, Klauder and Skagerstam proved and applied a Berezin-Lieb type inequality for their extended Husimi and Glauber-Sudarshan representations $S_{-\sigma}$ and $S_{\sigma}$. Let $\sigma,\sigma'\in \tco$ be positive operators with $\tr(\sigma)=\tr(\sigma')=1$. For an operator $A\in \tco$ and $\beta\in \reals$, Klauder and Skagerstam established that
 \begin{equation} \label{eq:blks}
  \iint_{\reals^{2d}} e^{-\beta A_{\sigma'}(z)} \ dz \leq \tr(e^{-\beta A}) \leq \iint_{\reals^{2d}} e^{-\beta A_{-\sigma}(z)} \ dz.
\end{equation}
By theorem \ref{thm:main}, one might expect that this result can be formulated using Werner's convolutions. In fact, Werner proved this result already in \cite{Werner:1984} in a more general form. Werner's proof uses the following three properties of the convolutions.

\begin{lem} \label{lem:poscor}
	Let $S$ be a positive trace class operator with $\tr(S)=1$, and consider $f\in L^{\infty}(\reals^{2d})$ and $T\in \bo$ with $f\geq 0$ and $T\geq 0$. Then
	\begin{enumerate}
		\item $f\ast S\geq 0$ and $T\ast S\geq 0$.
		\item $\tr(f\ast S)=\iint_{\reals^{2d}} f(z) \ dz$ and $\iint_{\reals^{2d}} T\ast S(z) \ dz = \tr(T)$.
		\item If $1$ denotes the function $1(z)=1$ on $\reals^{2d}$ and $\mathcal{I}$ is the identity operator on $\HS$, then $1\ast S=\mathcal{I}$ and $\mathcal{I}\ast S=1$.
	\end{enumerate}
\end{lem} 
\begin{proof}
	\begin{enumerate}
		\item The straightforward proof is left to the reader.
		\item Consider first $\tr(f\ast S)$. When $f\in L^1(\reals^{2d})$, we can use that $\F_W(f\ast S)(z)=\F(f)(z)\F_W(S)(z)$. Applying this at $z=0$ gives the result. To prove the result when $\int_{\reals^{2d}} f(z) \ dz=\infty$, one can approximate $f$ from below by functions in $L^1(\reals^{2d})$.
		
The second part is lemma \ref{lem:werner} when $T\in \tco$, and when $\tr(T)=\infty$ one can approximate $T$ by trace class operators to prove the result. 
		\item The convolution of $f\in L^{\infty}(\reals^{2d})$ with $S\in \tco$ is defined by duality, using the condition $\inner{f\ast S}{T}=\inner{f}{T\ast \check{S}^*}$ for any $T\in \tco$ \cite{Luef:2017vs}. One easily checks using the definitions and lemma \ref{lem:werner} that $\inner{1}{T\ast \check{S}^*}=\iint_{z\in \reals^{2d}} \overline{\tr(T\alpha_z S^*)} \ dz=\overline{\tr(T)}=\inner{\mathcal{I}}{T}$, hence $1\ast S=\mathcal{I}$. That $\mathcal{I}\ast S=1$ is proved similarly.
	\end{enumerate}
\end{proof}
In words, convolution with a fixed operator $S$ as in lemma \ref{lem:poscor} preserves trace/integral, positivity and identity. This is the key to Werner's proof of the Berezin-Lieb inequality, which now follows in an elaborated version.
\begin{prop} 
Fix a positive trace class operator $S$ with $\tr(S)=1$, and let $T=T^*\in \bo$ and $f=f^*\in L^{\infty}(\reals^{2d})$. If $\Phi$ is a positive, convex and continuous function on a domain containing the spectrums of $T$ and $S\ast T$, then
\begin{equation} \label{eq:blgs} 
  \iint_{\reals^{2d}}\Phi\circ (S\ast T)(z) \ dz \leq \tr(\Phi(T)).
\end{equation}
Similarly, if $\Phi$ is a positive, convex and continuous function on a domain containing the spectrums of $f$ and $f\ast S$, then
\begin{equation} \label{eq:blh}
  \tr(\Phi(f \ast S)) \leq \iint_{\reals^{2d}}\Phi\circ f(z) \ dz.
\end{equation}
\end{prop}
\begin{proof}
	First we will explain how the proof may be reduced to the case $\Phi(t)=t_+,$ i.e. the function that returns the positive part of $t$. The reader may confirm that the set of functions $\Phi$ where equations \eqref{eq:blgs} and \eqref{eq:blh} hold is a convex cone and closed under taking the supremum. It is also closed under reflection $\Phi\mapsto \check{\Phi}$ and translations $\Phi \mapsto T_x\Phi$ for $x\in \reals$ -- these facts follow from the spectral calculus and that convolutions with $S$ preserve identity. Since $\Phi$ is assumed to be positive, convex and continuous on a compact set, it can be approximated uniformly by positive piecewise linear convex functions \cite[p. 35]{Niculescu:2006}. As is shown in \cite[Thm. 1.5.7]{Niculescu:2006}, any positive piecewise linear convex function can be written as a linear combination of translates and reflections of the function $t_+$ -- hence $\Phi$ can be approximated by such functions. By linearity it is therefore enough to prove the result for $t_+$.	
	 
	 We will restrict the rest of the proof to inequality \eqref{eq:blgs}, inequality \eqref{eq:blh} follows from a similar argument. Observe that by the spectral calculus, 
	\begin{equation*} 
  \tr(T_+)=\inf \{\tr(A):A\geq 0, A\geq T \},
\end{equation*}
since $0\leq T\leq A$ implies $\tr(T)\leq \tr(A)$. Now consider the following calculation:
\begin{align*}
  \inf \{\tr(A):A\geq 0, A\geq T \}&=\inf \left \{ \iint_{\reals^{2d}} A\ast S(z) \ dz:A\geq 0, A\geq T \right \} \\
  &\geq \inf \left \{ \iint_{\reals^{2d}} A\ast S(z) \ dz:A\ast S\geq 0, A\ast S\geq T\ast S \right \} \\
  &\geq \inf \left \{ \iint_{\reals^{2d}} g(z) \ dz:g\geq 0, g\geq T\ast S \right \} \\
  &= \iint_{\reals^{2d}} (T\ast S)_+ \ dz.
\end{align*}
The first equality is simply part (2) of lemma \ref{lem:poscor}. The two inequalities follow since we take the infinum of larger sets: in the first case this is true by part (1) of lemma \ref{lem:poscor}, and in the second case it is trivially true. The final equality follows by simple integration theory.
\end{proof}

To obtain Klauber and Skagerstam's result in equation \eqref{eq:blks}, we set $\Phi(t)=e^{-\beta x}$, $S=\check{\sigma}$ in \eqref{eq:blgs}, and $S=\sigma^*$, $f=T_{-\sigma}$ in \eqref{eq:blh}.

\section{Quantum harmonic analysis on phase space} \label{sec:tauberian}
When introducing the generalized Husimi and Glauber-Sudarshan representations, Klauder and Skagerstam assumed that the set of zeros of $\F_W(\sigma)$ was empty. In Werner's theory of convolutions, there is no reason to assume this in general. In fact, since we claim that $\F_W$ is a kind of Fourier transform, one might hope that imposing conditions on the zero set of $\F_W(\sigma)$ would lead to results analogous to Wiener's Tauberian theorem for functions. Such theorems have in fact been proved in \cite{Werner:1984} and \cite{Kiukas:2012gt}, from which we take the following results. A proof and similar statements may be found in \cite{Werner:1984, Kiukas:2012gt, Luef:2017vs}.

\begin{thm}
	\label{thm:wernerequiv}
Let $\sigma\in \tco$. 
	\begin{itemize}
\item [(a)] The following are equivalent.  
\begin{enumerate}[label=(a\arabic*)]
		\item The set $\{z\in \reals^{2d}:\F_W (\sigma)(z)=0\}$ is empty.
		\item If $f \in L^{\infty}(\mathbb{R}^{2d})$ and $f\ast \sigma=0$, then $f=0$.
		\item If $T \in \bo$ and $T \ast \sigma=0$, then $T=0$.
		\item $ L^1(\reals^{2d})\ast \sigma$ is dense in $\SC^1$.
	\end{enumerate}
\item [(b)] The following are equivalent.  
\begin{enumerate}[label=(b\arabic*)]
		\item The set $\{z\in \reals^{2d}:\F_W(\sigma)(z)=0\}$ has dense complement.
		\item If $f \in L^1(\mathbb{R}^{2d})$ and $f\ast \sigma=0$, then $f=0$.
		\item If $T\in \tco$ and $\sigma\ast T=0$, then $T=0$.
		\item $ L^{\infty}(\reals^{2d})\ast \sigma$ is weak* dense in $\bo$.
	\end{enumerate}
\end{itemize}
\end{thm}	
Using theorem \ref{thm:main} we can formulate this result using the generalized Husimi and Glauber-Sudarshan representations. 
\begin{cor} \label{cor:tauberianforhusimi}
Let $\sigma\in \tco$. 
	\begin{itemize}
\item [(a)] The following are equivalent.  
\begin{enumerate}[label=(a\arabic*)]
		\item The set $\{z\in \reals^{2d}:\F_W (\sigma)(z)=0\}$ is empty.
		\item If $T\in \bo$ has a generalized Glauber-Sudarshan symbol $T_{-\sigma} \in L^{\infty}(\mathbb{R}^{2d})$, then $T_{-\sigma}\in L^\infty(\reals^{2d})$ is unique.
		\item Any $T\in \bo$ is uniquely determined by its generalized Husimi representation. Hence, if we have $T \in \bo$ with $T_{\sigma}=0$, then $T=0$. 
		\item The set of $T\in \tco$, where $T_{-\sigma}\in L^1(\reals^{2d})$ exists, is dense in $\tco$.
	\end{enumerate}
\item [(b)] The following are equivalent.  
\begin{enumerate}[label=(b\arabic*)]
		\item The set $\{z\in \reals^{2d}:\F_W(\sigma)(z)=0\}$ has dense complement.
		\item  If $T\in \tco$ has a generalized Glauber-Sudarshan symbol $T_{-\sigma} \in L^1(\mathbb{R}^{2d})$, then $T_{-\sigma}\in L^1(\reals^{2d})$ is unique.
		\item Any $T\in \tco$ is uniquely determined by its generalized Husimi representation. Hence, if we have $T \in \tco$ with $T_{\sigma}=0$, then $T=0$. 
		\item The set of $T\in \bo$, where $T_{-\sigma}\in L^{\infty}(\reals^{2d})$ exists, is weak* dense in $\bo$.
	\end{enumerate}
\end{itemize}
\end{cor}
\begin{proof}
	Using theorem \ref{thm:main}, we see that each part of this corollary is the same part of theorem \ref{thm:wernerequiv}, written in the terminology of generalized Husimi and Glauber-Sudarshan representations.We also use that parts (a1) and (b1) are satisfied for $\sigma$ if and only if they are satisfied for $\check{\sigma}$ or $\sigma^*$, since a simple calculation shows that $\F_W(\check{\sigma})(z)=\F_W(\sigma)(-z)$ and $\F_W(\sigma^*)(z)=\overline{\F_W(\sigma)(-z)}$.
\end{proof}
\begin{rem}
	Klauder and Skagerstam also discuss the case where $\F_{W}(\sigma)$ has zeros\cite{Klauder:2007}, and arrive at results similar to corollary \ref{cor:tauberianforhusimi}. They do not, however, consider which $L^p$-space the generalized Glauber-Sudarshan symbol $T_{-\sigma}$ belongs to.
\end{rem}
In the interpretation of W\'odkiewicz discussed in section \ref{sec:operational}, the generalized Husimi function $S_{\sigma}$ of $S\in \tco$ takes into account the effect of the measurement device, which is assumed to be in the state $\sigma\in \tco$. In this interpretation the equivalence of (b1) and (b3) in the corollary means that if $\{z\in \reals^{2d}:\F_W(\sigma)(z)=0\}$ has dense complement, then the noise from a measurement device in the state $\sigma$ may in theory be removed from the distribution $S_{\sigma}$ to recover the state $S$. This connection and result has been noted in \cite{Kiukas:2012gt,Busch:2016}, who also noted that since there exist operators $\sigma$ such that $\{z\in \reals^{2d}:\F_W(\sigma)(z)=0\}$ does \textit{not} have dense complement, corollary \ref{cor:tauberianforhusimi} contradicts the claim in \cite{Buzek:1997} that the generalized Husimi function $S_{\sigma}$ always "contains complete information about the measured state" \cite[p. 636]{Buzek:1997}, and hence it also shows that a formal reconstruction in \cite{Buzek:1997} procedure to get $S$ from $S_{\sigma}$ will not always work. \\

Finally, we mention that the spectrum of an operator $\sigma$, defined as the closure of $\{z\in \reals^{2d}:\F_W(\sigma)(z)\neq 0\}$, may be identified with the \textit{Arveson spectrum} of $\sigma$ for the automorphism $\alpha$, see \cite{Luef:2017vs}. 
\section{Berezin quantization on other phase spaces}
Theorem \ref{thm:main} shows that the generalized Glauber-Sudarshan and Husimi representations may be expressed using the convolutions of operators and functions. Hence, we see that if we can define these convolutions on more general phase spaces, we would also obtain generalized Glauber-Sudarshan and Husimi representations on these phase spaces. Kiukas et. al \cite{Kiukas:2006} have shown that the integral defining $f\ast \sigma^*$ in equation \eqref{eq:operatorconv} exists in the more general context where $G$ is a locally compact, unimodular\footnote{Kiukas \cite{Kiukas:2006a} has also considered results of this kind for non-unimodular groups $G$.}, second countable, connected Hausdorff group  with Haar measure $\mu_G$. Let $U:G\to B(\mathcal{H})$ be a projective unitary representation of $G$ on some separable Hilbert space $\mathcal{H}$. 
\\
In this case, for any positive $\sigma\in \tco(\mathcal{H})$ with $\tr(\sigma)=1$ and $f\in L^{\infty}(G)$, the vector-valued integral 
\begin{equation} \label{eq:holevo}
  \iint_{G} f(g) U(g)\sigma^* U(g)^* \ d\mu_G(g)
\end{equation}
defines a bounded operator on $\mathcal{H}$ in the ultraweak sense \cite[Thm. 1]{Kiukas:2006} , i.e. the weak* topology on $B(\mathcal{H})$ induced by $B(\mathcal{H})$ being the dual space of $\tco(\mathcal{H})$. Since the generalized Glauber-Sudarshan symbol $S_{-\sigma}$ on $\reals^{2d}$ is given by $S=S_{-\sigma}\ast \sigma^*$ by  theorem \ref{thm:main}, we see that on a group $G$ as above we may consider $f$ to be a Glauber-Sudarshan symbol of $\iint_{G} f(g) U(g)\sigma^* U(g)^* \ d\mu_G(g)$. 
\\
Since Berezin quantization on $\reals^{2d}$ may also be expressed as the convolution operation $f\mapsto f\ast \sigma^*$ with $\sigma=\phi\otimes \phi$ where $\phi(x)=2^{d/4}e^{-\pi x\cdot x}$, we see that each $\sigma\in \tco(\mathcal{H})$ defines a version of Berezin quantization on $G$ by 
\begin{equation*}
  f\mapsto \iint_{G} f(g) U(g)\sigma^* U(g)^* \ d\mu_G(g) \ \text{ for } f\in L^{\infty}(G).
\end{equation*}
  Similarly we may define a generalized Husimi representation $S_{\sigma}$ of $S\in \tco(\mathcal{H})$ on $G$ by $S_{\sigma}(g)=\tr(SU(g)\sigma U(g)^{\ast})$, simply by replacing $\pi(z)$ with $U(g)$ in equation \eqref{eq:convolutionoftwooperators} which defines the convolution of two operators. 
	This Husimi representation will then satisfy an integrability condition similar to lemma \ref{lem:werner} \cite[Lem. 2]{Kiukas:2006}, provided that $G$ is unimodular \cite{Kiukas:2006a}. \\These considerations may also be formulated in terms of covariant observables\cite{Werner:1984,Holevo:1979,Busch:2016,Cassinelli:2003,Kiukas:2006,Kiukas:2006a}; a special class of positive operator valued measures. In this setting Holevo \cite{Holevo:1979} and later others \cite{Werner:1984,Kiukas:2006a,Kiukas:2006,Cassinelli:2003} proved a "converse" of equation \eqref{eq:holevo} that amounts to a characterization of covariant observables; namely that any map $\Gamma:L^{\infty}(G)\to B(\mathcal{H})$ satisfying certain properties must be given by $f\mapsto \iint_{G} f(g) U(g)\sigma^* U(g)^* \ d\mu_G(g)$ for some positive $\sigma\in \tco(\mathcal{H})$ with $\tr(\sigma)=1$.  A reference for these topics is the book \cite{Busch:2016} and references therein.

\end{document}